\documentclass[a4wide,11pt]{article}
\usepackage{a4wide}
\usepackage{amsmath}
\usepackage{hyperref}
\usepackage{amsthm}
\usepackage{array}
\usepackage{alltt}
\usepackage{mathtools}
\usepackage{stmaryrd}
\usepackage{amssymb}
\usepackage{tikz}
\usepackage[linesnumbered]{algorithm2e}
\usepackage{cleveref}
\usepackage{subcaption}
\usepackage[affil-it]{authblk}

\usetikzlibrary{automata,positioning,decorations.pathreplacing}

\newcommand{\bisim}{\leftrightarroweq}

\newcommand{\bigO}{\mathcal{O}}
\newcommand{\mfalse}{\mathrm{false}} 
\newcommand{\mtrue}{\mathrm{true}} 
\newcommand{\BCRP}{BCRP}

\newtheorem{thm}{Theorem} 
\newtheorem{definition}{Definition}
\newtheorem{lemma}{Lemma}

\newtheorem{fact}{Fact}

\newcommand{\pijl}[1]{\mathrel{\text{$\xrightarrow{\smash[t]{#1}}$}}}

\newcommand{\Act}{\mathit{Act}}
\newcommand{\Nat}{{\mathbb{N}}}
\newcommand{\rel}{{\rightarrow}}
\newcommand{\Bool}{{\mathbb{B}}}
\newcommand{\true}{{\mathrm{true}}}
\newcommand{\false}{{\mathrm{false}}}
\newcommand{\blocklabel}{L_B}

\newcommand{\partit}{\pi}
\newcommand{\markoffset}{\ensuremath{\mathit{off}}}

\tikzset{  plate/.style={draw, shape=rectangle, rounded corners=0.5ex, thick, align=right, inner sep=10pt, inner ysep=10pt,label={[xshift=-20pt,yshift=15pt]south east:#1}}}

\title{\textbf{A linear parallel algorithm to compute bisimulation and relational coarsest partitions}}

\author[1]{Jan Martens}
\author[1]{Jan Friso Groote}
\author[1]{Lars van den Haak}
\author[1,2]{\\Pieter Hijma}
\author[1]{Anton Wijs}

\affil[1]{Eindhoven University of Technology, The Netherlands}

\affil[ ]{\textsf{\small\{j.j.m.martens, j.f.groote, l.b.v.d.haak, a.j.wijs\}@tue.nl}}
\affil[2]{Vrije Universiteit Amsterdam, The Netherlands}
\affil[ ]{\textsf{\small pieter@cs.vu.nl}}
\begin{document}
	\maketitle              
\begin{abstract}
\noindent%
The most efficient way to calculate strong bisimilarity is by finding the relational coarsest partition on a transition system. 
We provide the first linear time algorithm to calculate strong bisimulation 
using parallel random access machines (PRAMs). 
More precisely, with $n$ states, $m$ transitions and $|\mathit{Act}|\leq m$ action labels, we provide an algorithm on $\max(n,m)$ processors
that calculates strong bisimulation in time $\bigO(n+|\mathit{Act}|)$ and space $\bigO(n+m)$. 
The best-known PRAM algorithm
has time complexity $\bigO(n\log n)$ on a smaller number of processors 
making it less suitable for massive parallel devices such as GPUs.
An implementation on a GPU shows that the linear time-bound is achievable on contemporary hardware. 


	
\end{abstract}

\section{Introduction}
The notion of \textit{bisimilarity} for Kripke structures and Labeled Transition Systems (LTSs) is commonly used to define behavioural equivalence. Deciding this behavioural equivalence is important in the field of modelling and verifying concurrent systems~\cite{mcrl2,milner1980calculus}. Kanellakis and Smolka proposed a partition refinement algorithm for obtaining the bisimilarity relation for Kripke structures~\cite{kanellakis1983}. The proposed algorithm has a run time complexity of $\bigO(nm)$ where $n$ is the number of states and $m$ is the number of transitions of the input graph. Later, a more sophisticated refinement algorithm running in $\bigO(m~log~n)$ steps was proposed by Paige and Tarjan~\cite{paige1987three}. 

In recent years the increase in the speed of sequential chip design has stagnated due to a multitude of factors such as energy consumption and heat generation. In contrast, parallel devices such as graphics processing units (GPUs) keep increasing rapidly in computational power. In order to profit from the acceleration of these devices, we require algorithms with massive parallelism.  The article ``There's plenty of room at the Top: What will drive computer performance after Moore's law'' by Leierson et al.~\cite{Leisersoneaam9744} 
indicates that the advance in computational performance will come from software and algorithms that
can employ hardware structures with a massive number of simple, parallel processors, such as
GPUs. In this paper, we propose such an algorithm to decide bisimilarity. 


Deciding bisimilarity is $P$-complete~\cite{balcazar_1992}, which suggests that bisimilarity is an inherently sequential problem. This fact has not withheld the community from searching efficient parallel algorithms for deciding bisimilarity of Kripke structures. In particular, Lee and Rajasekaran~\cite{lee_1994_RCPP} proposed a parallel algorithm based on the Paige Tarjan algorithm that works in $\bigO(n\ log\ n)$ time complexity using $\frac{m}{\log n}\log \log n$ Concurrently Read and Concurrently Write (CRCW) processors.

In this work, we improve on the best known theoretical bound for PRAM algorithms using a higher degree of parallelism. The proposed algorithm improves the run time complexity to $\bigO(n)$ on $max(m,n)$ processors and is based on the sequential algorithm of Kanellakis and Smolka~\cite{kanellakis1983}. The larger number of processors used in this algorithm favours the increasingly parallel design of contemporary and future hardware. In addition, the algorithm is \textit{optimal} w.r.t.\ the sequential Kanellakis-Smolka algorithm, meaning that overall, it does not perform more work than its sequential counterpart.

We first present our algorithm on Kripke structures where transitions are unlabelled. 
However, as labelled transition systems 
(LTSs) are commonly used, and labels are not straightforward to incorporate in an efficient way (cf.\ for instance \cite{DBLP:journals/fuin/Valmari10}),
we discuss how our algorithm can be extended to take action labels into account. This leads to an algorithm with a run time 
complexity of $\bigO(n + |\Act|)$, with $\Act$ the set of action labels.

Our algorithm has been designed for and can be analyzed with the CRCW PRAM model, following notations from~\cite{stockmeyer1984simulation}. This model is an extension of the normal RAM model, allowing multiple processors to work with shared memory. In the CRCW PRAM model, parallel algorithms can be described in a straightforward and elegant way.
In reality, no device exists that completely adheres to this PRAM model, but with recent advancements, hardware gets better and better at approximating the model since the number of parallel threads keeps growing. We demonstrate this by translating the PRAM algorithm
to GPU code. 
We straightforwardly implemented our algorithm in CUDA and experimented with an NVIDIA Titan RTX, showing that our algorithm 
performs mostly in line with what our PRAM algorithm predicts.

The paper is structured as follows: In Section~\ref{sec:prelims}, we recall the necessary preliminaries on the CRCW PRAM model and state the partition refinement problems this paper focuses on. In Section~\ref{sec:algorithm}, we propose a parallel algorithm to compute bisimulation for Kripke structures, which is also called the Relational Coarsest Partition Problem (RCPP). In this section, we also prove the correctness of the algorithm and provide a complexity analysis. In Section~\ref{sec:labels}, we discuss the details for an implementation with multiple action labels, thereby supporting LTSs, which forms the Bisimulation Coarsest Refinement Problem (\BCRP{}). In Section~\ref{sec:experiments} we discuss the results of the implementation and in Section~\ref{sec:weaker} we address the usage of weaker PRAM models. Finally, in Section~\ref{sec:relatedwork}, we discuss related work. 
\section{Preliminaries}\label{sec:prelims}
\subsection{The PRAM Model}
The \textit{Parallel Random Access Machine} (PRAM) is a natural extension of the normal Random Access Machine (RAM), where an arbitrary number of parallel programs can access the memory. Following the definitions of~\cite{stockmeyer1984simulation} we use a version of PRAM that is able to Concurrently Read and Concurrently Write (CRCW PRAM). It differs from the model introduced in~\cite{wyllie1978parallelism} in which the PRAM model was only allowed to concurrently read from the same memory address, but concurrent writes (to the same address) could not happen. We call the model from~\cite{wyllie1978parallelism} an Concurrent Read, Exclusive Write (CREW) PRAM model.

A CRCW PRAM consists of a sequence of numbered processors $P_0, P_1, \dots$. These processors have all the natural instructions of a normal RAM such as addition, subtraction, and conditional branching based on the equality and less-than operators. There is an infinite amount of common memory the processors have access to. The processors have instructions to read from and write to the common memory. In addition, a processor $P_i$ has an instruction to obtain its unique index $i$. A PRAM also has a function $P:\Nat \to \Nat$ which defines a bound on the number of processors given the size of the input.

All the processors have the same program and run synchronized in a single instruction, multiple data (SIMD) fashion. In other words, all processors execute the program in lock-step. Parallelism is achieved by distributing the data elements over the processors and having the processors apply the program instructions on `their' data elements. 

Initially, given input consisting of $n$ data elements, the CRCW PRAM assumes that the input is stored in the first $n$ registers of the common memory, and starts the first $P(n)$ processors $P_0, P_1, \dots ,P_{P(n)-1}$. 

We need to define what the behaviour of the machine will be whenever a concurrent write happens. The way to handle this memory contention in concurrent writes is usually by assuming one of the following:
\begin{itemize}
	\item \textbf{(Common)} All processors try to write the same value and succeed, otherwise, the writes are not legal and fail;
	\item \textbf{(Arbitrary)} Only one arbitrary attempt to write succeeds; 
	\item \textbf{(Priority)} Only the processor with the lowest index succeeds in writing.
\end{itemize}

The algorithm proposed in this paper works if we make either the \textit{arbitrary} or the \textit{priority} assumption. In Section \ref{sec:weaker} we explain how we can
adapt it to work under the \textit{common} assumption.

A  parallel program for a PRAM is called \textit{optimal} w.r.t.\ a sequential algorithm if the total work done by the program does not exceed the work done by the sequential algorithm. More precisely, if $T$ is the parallel run time and $P$ the number of processors used, then the algorithm is optimal w.r.t.\ a sequential algorithm running in $S$ steps if $P\cdot T \in \bigO(S)$.

The computational complexity of these models is well studied and there is a close relation between circuit complexity and the complexity of PRAM algorithms~\cite{stockmeyer1984simulation}.
\subsection{Strong Bisimulation}

To formalise concurrent system behaviour, we use LTSs.

\begin{definition}[Labeled Transition System] A Labeled Transition System (LTS) is a three-tuple $A=(S, Act,  \rightarrow)$ where $S$ is a finite set of states, $Act$ a finite set of action labels, and $\rightarrow \subseteq S\times Act \times S$ the transition relation. 
\end{definition}
Let $A = (S, Act, \rel)$ be an LTS. Then, for any two states $s,t\in S$ and $a \in Act$, we write $s\pijl{a} t$ iff $(s,a,t)\in \rel$. 

Kripke structures differ from LTSs in the fact that the states are labelled as opposed to the transitions. In the current paper, for convenience, instead of using Kripke structures where appropriate, we reason about LTSs with a single action label, i.e., $|\Act| = 1$. Computing the coarsest partition of such an LTS can be done in the same way as for Kripke structures, apart from the fact that in the latter case, a different initial partition is computed that is based on the state labels (see, for instance, \cite{jansen_2020}).

\begin{definition}[Strong bisimulation]
\label{def:bisim}
On an LTS $A = (S,Act, \rel)$ a relation $R\subseteq S\times S$ is called a strong bisimulation relation if and only if it is symmetric and for all $s,t\in S$ with $s R t$ and for all $a\in Act$ with $s\pijl{a} s'$, we have:
	
		$$\exists t' \in S. t\pijl{a} t' \wedge s' R t'$$   
\end{definition}

Whenever we refer to bisimulation we mean strong bisimulation. Two states $s, t \in S$ in an LTS $A$ are called \emph{bisimilar}, denoted by $s \bisim t$, iff there is some bisimulation 
relation $R$ for $A$ that relates $s$ and $t$. 

A \textit{partition} $\partit$ of a finite set of states $S$ is a set of subsets that are pairwise disjoint and whose union is equal to $S$, i.e., $\bigcup_{B\in \partit} B = S$. Every element $B\in \partit$ of this partition $\partit$ is called a \textit{block}. 

We call partition $\partit'$ a \textit{refinement} of $\partit$ iff for every block $B'\in \partit'$ there is a block $B \in \partit$ such that $B' \subseteq B$. We say a partition $\partit$ of a finite set $S$ induces the relation $R = \{(s,t) \mid \exists B \in \partit . s \in B \wedge t \in B \}$. This is an equivalence relation of which the blocks of $\partit$ are the equivalence classes. 

Given an LTS $A = (S, Act, \rel)$ and two states $s,t\in S$ we say that $s$ \textit{reaches} $t$ with action $a\in Act$ iff $s\pijl{a}t$. 
A state $s$ \textit{reaches} a set $U\subseteq S$ with an action $a$ iff there is a state $t\in U$ such that $s$ reaches $t$ with action $a$. 
A set of states $V\subseteq S$ is called \textit{stable} under a set of states 
$U\subseteq S$
iff for all actions $a$ either all states in $V$ reach $U$ with $a$, or no state in $V$
reaches $U$ with $a$. A partition $\partit$ is stable under a set of states $U$ iff
each block $B\in\partit $ is stable under $U$. The partition $\partit$ is called stable iff it is stable under all its own blocks $B\in \partit$. 

\begin{fact}\label{fact:stable}\cite{paige1987three}
Stability is inherited under refinement, i.e. given a partition $\partit$ of $S$ and a refinement $\partit'$ of $\partit$, then if $\partit$ is stable under $U\subseteq S$,
then $\partit'$ is also stable under $U$.
\end{fact}

The main problem we focus on in this work is called the bisimulation refinement problem (\textbf{\BCRP{}}). It is defined as follows:

\textbf{Input:}  An LTS $M = (S,Act, \rel)$.

\textbf{Output:} The partition $\partit$ of $S$ which is the coarsest partition, i.e., has the smallest number of blocks, that forms a bisimulation relation.

In a Kripke structure, the transition relation forms a single binary relation, since the transitions are unlabelled. This is also the case when an LTS has a single action label. In that case, the problem is called the Relational Coarsest Partition Problem (\textbf{RCPP})~\cite{kanellakis1983,lee_1994_RCPP,paige1987three}. This problem is defined as follows:

\textbf{Input:} A set $S$, a binary relation $\rightarrow: S\times S$ and an initial partition $\pi_0$

\textbf{Output:} The partition $\pi$ which is the coarsest refinement of $\pi_0$ and which is a bisimulation relation.

It is known that \BCRP{} is not significantly harder than RCPP as there are intuitive translations from LTSs to Kripke structures~\cite[Dfn. 4.1]{de1990action}. 
However, some non-trivial modifications can speed-up the algorithm for some cases, hence we discuss both problems separately. In Section \ref{sec:algorithm}, we discuss the basic parallel algorithm for RCPP, and in Section~\ref{sec:labels}, we discuss the modifications required to efficiently solve the \BCRP{} problem for LTSs with multiple action labels.

\section{Relational Coarsest Partition Problem}\label{sec:algorithm}
\subsection{A Sequential Algorithm}
In this section, we discuss a sequential algorithm based on one of Kanellakis and Smolka~\cite{kanellakis1983} for RCPP. This is the basic algorithm which we adapt to the parallel PRAM algorithm. The algorithm starts with an input partition $\partit_0$ and refines all blocks until a stable partition is reached. This stable partition will be the coarsest refinement that defines a bisimulation relation.

The sequential algorithm, Algorithm~\ref{lst:KS}, works as follows. Given are a set $S$, 
a relation $\rightarrow\subseteq S\times S$, and an initial partition $\pi_0$ of $S$. 
Initially, we mark the partition as not necessarily stable under all blocks by putting
these blocks in a set $\mathit{Unstable}$. In any iteration of the algorithm, if a block $B$ of the current partition is not in
$\mathit{Unstable}$, then the current partition is stable
under $B$. If $\mathit{Unstable}$ is empty, the partition is stable under all its blocks,
and the partition represents the required bisimulation. 

As long as some blocks are in $\mathit{Unstable}$ (line 3), a single block $B\in\pi$ is taken from this set (line 4) and we split the current partition such that it becomes stable under $B$. Therefore, we refer to this block as the \emph{splitter}. The set $S' = \{s \in S \mid \exists t \in B. s\rightarrow t \}$ is the reverse image of $B$ (line 6). This set consists of all states that can reach $B$, and we use $S'$ to define our new blocks. All blocks $B'$ that have a non-empty intersection with $S'$, i.e., $B' \cap S' \neq \emptyset$, and are not a subset of $S'$, i.e., $B'\cap S' \neq B'$ (line 7), are split in the subset of states that reach $S'$ and the subset of states that do not reach $S'$ (lines 8-9). These two new blocks are added to the set of $\mathit{Unstable}$ blocks (line 10). The number of states is finite, and blocks can be split only a finite number of times. Hence, blocks are only finitely often put
in $\mathit{Unstable}$, and so the algorithm terminates. 

\begin{algorithm}[t]
	$\partit:= \partit_0$\;
	$\mathit{Unstable} := \pi$\;
	\While{$\mathit{Unstable} \neq \emptyset$}{
		\ForEach{$B\in \mathit{Unstable}$}{
			$\mathit{Unstable} := \mathit{Unstable} \setminus \{B\}$\;
			$S' := \{s \in S \mid \exists t \in B. s\pijl{} t \}$\;				
			\ForEach{$B'\in \partit \textit{ with } \emptyset \subset B' \cap S' \subset B'$}{
				\tcp{Split $B'$ into $B' \cap S'$ and $B' \setminus S'$}
				$\partit := \partit \setminus \{B\}$\;
				$\partit:= \partit \cup \{B'\cap S', B' \setminus S'\}$\;
				$\mathit{Unstable} := \mathit{Unstable} \cup \{B'\cap S', B' \setminus S'\}$\;
			}
		}
	}
	\caption{\label{lst:KS}Sequential algorithm based on Kanellakis-Smolka}
\end{algorithm}

\subsection{The PRAM Algorithm}

Next, we describe a PRAM algorithm to solve RCPP that is based on the sequential algorithm given in Algorithm~\ref{lst:KS}.

\subsubsection{Block representation}
Given an LTS $A = (S, Act, \rightarrow)$ with $|A| = 1$ and $|S| = n$ states, we assume that the states are labeled with unique indices $0, \dots, n-1$. A partition $\pi$ in the PRAM algorithm is represented by assigning a block label from a set of block labels $\blocklabel$ to every state. The number of blocks can never be larger than the number of states, hence,
we use the indices of the states as block labels: $L_B=S$. We exploit this in the PRAM algorithm to efficiently select a new block label whenever a new block is created. We select the block label of a new block by electing one of its states to be the \textit{leader} of that block and using the index of that state as the block label. By doing so, we maintain an invariant that the leader of a block is also a member of the block. 
%

In a partition $\partit$, whenever a block $B\in \partit$ is split into two blocks $B'$ and $B''$, the leader $s$ of $B$ which is part of $B'$ becomes the leader of $B'$, and for $B''$, a new state $t \in B''$ is elected to be the leader of this new block. Since the new leader is not part of any other block, the label of $t$ is fresh with respect to the block labels that are used for the other blocks. This method of using state leaders to represent subsets was first proposed in~\cite{wijsscc,wijs_2015}.

\subsubsection{Data structures}
The common memory contains the following information:
\begin{enumerate}
	\item $n:\Nat$, the number of states of the input.
	\item $m:\Nat$, the number of transitions of the input relation.
	\item The input, a fixed numbered list of transitions. For every index $0\leq i<m$ of a transition, a source $\textit{source}_i\in S$ and target $\textit{target}_i\in S$ are given, representing the transition $\textit{source}_i\to \textit{target}_i$.
	\item $C: \blocklabel \cup \{\bot\}$, the label of the current block that is used as a splitter; $\bot$ indicates that no splitter has been selected.
	\item The following is stored in lists of size $n$, for each state with index $i$:
	\begin{enumerate}
		\item $\mathit{mark}_i: \Bool$, a mark indicating whether state $i$ is able to reach the splitter.
		\item $\mathit{block}_i:\blocklabel$, the block of which state $i$ is a member.
	\end{enumerate}
	\item The following is stored in lists of size $n$, for each potential block with block label $i$:
	\begin{enumerate}
		\item $\mathit{new\_leader}_i : \blocklabel$ the leader of the new block when a split is performed. 
		\item $\mathit{unstable}_i : \Bool$ indicating whether $\partit$ is possibly unstable w.r.t.\ the block. 
	\end{enumerate}
\end{enumerate}

As input, we assume that each state with index $i$ has an input variable $I_i\in L_B$ that is the initial block label. In other words, the values of the $I_i$ variables together encode $\pi_0$. Using this input, the initial values of the block label $\textit{block}_i$ variables are calculated to conform to our block representation with leaders. Furthermore in the initialization, $\textit{unstable}_i = \mfalse$ for all $i$ that are not used as block label, and $\mtrue$ otherwise.

\subsubsection{The algorithm} We provide our first PRAM algorithm in Algorithm~\ref{lst:alg}. The PRAM is started with $max(m,n)$ processors. These processors are dually used for transitions and states.

The algorithm performs initialisation (lines 1-6), after which each block has selected a new leader (lines 3-4), ensuring that the leader is one of its own states, and the initial blocks are set to unstable. Subsequently, the algorithm enters a single loop that can be explained in three separate parts. 

\begin{description}
	\item[Splitter selection (lines~8-14), executed by $n$ processors.] Every variable $mark_i$ is set to $\mfalse$. After this, every processor with index $i$ will check $\mathit{unstable}_i$. If block $i$ is marked unstable the processor tries to write $i$ in the variable $C$. If multiple write accesses to $C$ happen concurrently in this iteration, then according to both the arbitrary and the priority PRAM model (see Section~\ref{sec:prelims}), only one process $j$ will succeed in writing, setting $C:=j$ as splitter in this iteration.
	
	\item[Mark states (lines~15-17), executed by $m$ processors.]  Every processor $i$ is responsible for the transition $s_i\pijl{} t_i$ and checks if $t_i$ ($\mathit{target}_i$) is in the current block $C$ (line~\ref{alg:reachable}). If this is the case the processor writes $\mtrue$ to $mark_{\mathit{source}_i}$ where $\mathit{source}_i$ is $s_i$.  
	This mark now indicates that $s_i$ reaches block $C$.
	
	\item[Performing splits (lines~18-26), executed by $n$ processors.] Every processor $i$ compares the mark of state $i$, i.e., $\mathit{mark}_i$, with the mark of the leader of the block in which state $i$ resides, i.e., $\mathit{mark}_{\mathit{block}_i}$ (line~\ref{alg:shouldsplit}). If the marking is different, state $i$ has to be split from $\mathit{block}_i$ into a new block. At Line~\ref{alg:leader}, a new leader is elected among the states that form the newly created block. The index of this leader is stored in $\mathit{new\_leader}_{\mathit{block}_i}$. The unstability of block $\mathit{block}_i$ is set to $\mtrue$ (line 22). After that, all involved processors update the block index for their state (line 21) and update the stability of the new block (line 22).
\end{description}

\begin{figure*}[tb]
	\centering
	\scalebox{1}{
	\begin{tikzpicture}[->,shorten >=1pt,auto,node distance=1.25cm]
		
		\node[state] (A)                    {$s_1$};
		\node[state]         (B) [below right = 1.5cm and 0.1cm of A] {$s_4$};
		\node[state]         (C) [right of=A] {$s_2$};
		\node[state]         (D) [right of = B] {$s_5$};
		\node[state]         (E) [right of=C] {$s_3$};
		
		\path (A) edge (B)
		(C) edge (D)
		(C) edge[bend right] (E)
		(E) edge[bend right] (C)
		;

		\node[plate=$B_{s_1}$ , minimum width=4cm, minimum height=1.5cm, rectangle,draw, anchor=west, xshift=-0.25cm] (p1) at (A.west){};
		\node[plate=$B_{s_4}$, minimum width=2.75cm, minimum height=1.5cm, rectangle,draw, anchor=west, xshift=-0.25cm, color=blue] (p2) at (B.west) {};
		
		\node (text) [below of=p2] {Step 1: Select \textit{current\_block:= $B_{s_4}$}};
		
		\node[state, color=blue] (A1)[right= 1.5 of E]                    {$s_1$};
		\node[state]         (B1) [below right= 1.5cm and 0.1cm of A1] {$s_4$};
		\node[state, color=blue]         (C1) [right of=A1] {$s_2$};
		\node[state]         (D1) [right of = B1] {$s_5$};
		\node[state]         (E1) [right of=C1] {$s_3$};
		
		\path (A1) edge[color=blue] (B1)
		(C1) edge[color=blue] (D1)
		(C1) edge[bend right] (E1)
		(E1) edge[bend right] (C1)
		;

		\node[plate=$B_{s_1}$ , minimum width=4cm, minimum height=1.5cm, rectangle,draw, anchor=west, xshift=-0.25cm] (p11) at (A1.west){};
		
		\node[plate=$B_{s_4}$, minimum width=2.75cm, minimum height=1.5cm, rectangle,draw, anchor=west, xshift=-0.25cm] (p12) at (B1.west) {};
		
		\node (text1) [below of=p12] {Step 2: Mark nodes $s_1,s_2$};
		
		\node[state] (A2)[right= 1.5 of E1]                    {$s_1$};
		\node[state]         (B2) [below right= 1.5cm and 0.1cm of A2] {$s_4$};
		\node[state]         (C2) [right of=A2] {$s_2$};
		\node[state]         (D2) [right of = B2] {$s_5$};
		\node[state]         (E2) [right= 1cm of C2] {$s_3$};
		
		\path (A2) edge (B2)
		(C2) edge (D2)
		(C2) edge[bend right] (E2)
		(E2) edge[bend right] (C2)
		;
		
		\node[plate=$B_{s_1}$ , minimum width=2.75cm, minimum height=1.5cm, rectangle,draw, anchor=west, xshift=-0.25cm,color=blue] (p21) at (A2.west){};
		\node[plate=$B_{s_3}$ , minimum width=1.5cm, minimum height=1.5cm, rectangle,draw, anchor=west, xshift=-0.25cm, color=blue] (p21) at (E2.west){};
		\node[plate=$B_{s_4}$, minimum width=2.75cm, minimum height=1.5cm, rectangle,draw, anchor=west, xshift=-0.25cm] (p22) at (B2.west) {};
		
		\node (text2) [below of=p22] {Step 3: Split $B_{s_1}$ into $B_{s_1}, B_{s_3}$};
	\end{tikzpicture}}
	\caption{One iteration of Algorithm~\ref{lst:alg}\label{fig:iteration}}
\end{figure*}

The steps of the program are illustrated in Figure~\ref{fig:iteration}. The notation $B_{s_i}$ refers to a block containing all states that have state $s_i$ as their block leader. In the figure on the left, we have two blocks $B_{s_1}$ and $B_{s_4}$, of which at least $B_{s_4}$ is marked unstable. Block $B_{s_4}$ is selected to be splitter, i.e., $C = B_{s_4}$ at line 12 of Algorithm~\ref{lst:alg}. In the figure in the middle, $\mathit{mark}_i$ is set to $\mtrue$ for each state $i$ that can reach $B_{s_4}$ (line 16). Finally, block $B_{s_4}$ is set to stable (line 19), all states compare their mark with the leader's mark, and the processor working on state $s_3$ discovers that the mark of $s_3$ is different from the mark of $s_1$, so $s_3$ is elected as leader of the new block $B_{s_3}$ at line 21 of Algorithm~\ref{lst:alg}. Both $B_{s_1}$ and $B_{s_3}$ are set to unstable (lines 22 and 24).

The algorithm repeats execution of the \textbf{while}-loop until all blocks are marked stable. 

\begin{algorithm}[t]
	\SetAlgoLined	
	\If{$i < n$}{
		\tcp{Initialize all variables}
		$\mathit{unstable}_i := \mfalse$\;
		$\mathit{new\_leader}_{I_i} := i$\;
		$\mathit{block}_i := new\_leader_{I_i}$\;
		$\mathit{unstable}_{\mathit{block}_i} := \mtrue$\;
	}
	\SetKwRepeat{Do}{do}{while}
	\Do{$C \neq \bot$}{
		$C := \bot$\;
		\If{$i < n$}{
			$\mathit{mark}_i := \mfalse$\; 
			\If{$\mathit{unstable}_i$}{
				$C := i$\;\label{alg:block}
			}
		}
		\If{$i < m$ and $\mathit{block}_{\mathit{target_i}} = C$\label{alg:reachable}} {
			$\mathit{mark}_{\mathit{source_i}} := \mtrue$\;
		}
		\If{$i < n$ and $C \neq \bot$}{
			$\mathit{unstable}_{C} := \mfalse$\;\label{alg:stability_c}
			\If{$\mathit{mark}_i \neq \mathit{mark}_{\mathit{block}_i}$\label{alg:shouldsplit}}{
				$\mathit{new\_leader}_{\mathit{block}_i} := i$\;\label{alg:leader}
				$\mathit{unstable}_{\mathit{block}_i} := \mtrue$\;
				$\mathit{block}_i := \mathit{new\_leader}_{\mathit{block}_i}$\; 
				$\mathit{unstable}_{\mathit{block}_i} := \mtrue$\;						
			}
			
		}
	}
	\caption{\label{lst:alg}The algorithm for each processor $P_{i}$ in the PRAM with $i\in[0, \dots, max(n,m)]$}
\end{algorithm}

\subsection{Correctness}
The $\mathit{block}_i$ list in the common memory at the start of iteration $k$ defines a partition $\partit_k$ where states $s\in S$ with equal block labels \textit{block$_i$} form the blocks:

$$\pi_k = \{\{s \in S \mid block_s = s'\}\mid s'\in S\} \setminus \emptyset$$

A run of the program produces a sequence $\partit_0, \partit_1, \dots$ of partitions. Observe that partition $\partit_k$ is a refinement of every partition $\partit_0,\partit_1,\dots, \partit_{k-1}$, since blocks are only split and never merged. 

A partition $\partit$ induces a relation of which the blocks are the equivalence classes. For an input partition $\pi_0$ we call the relation induced by the coarsest refinement of $\pi_0$ that is a bisimulation relation $\bisim_{\pi_0}$.

We now prove that Algorithm~\ref{lst:alg} indeed solves RCPP. We first introduce Lemma~\ref{lemma:invariant} which is invariant throughout execution and expresses that states which are related by $\bisim_{\pi_0}$ are never split into different blocks. This lemma implies that if a refinement forms a bisimulation relation, it is the coarsest.

\begin{lemma}\label{lemma:invariant}
	Let $S$ be the input set of states, $\rightarrow:S\times S$ the input relation and $\pi_0$ the input partition. Let $ \partit_1,\partit_2, \dots$ be the sequence of partitions produced by Algorithm~\ref{lst:alg}, then for all initial blocks $B\in \partit_0$, states $s,t\in B$ and iteration $k\in \Nat$:
	$$ s\bisim_{\pi_0} t \implies \exists B\in \partit_k . s,t\in B$$
\end{lemma}
\begin{proof}
	This is proven by induction on $k$. In the base case, $\partit_0$, this is true by default. Now assume for a particular $k\in \Nat$ that the property holds. We know that the partition $\partit_{k+1}$ is obtained by splitting with respect to a block $C\in \partit_k$. For two states $s,t\in S$ with $s\bisim_{\pi_0} t$ we know that $s$ and $t$ are in the same block in $\partit_k$. In the case that both $s$ and $t$ do not reach $C$, then $mark_s = mark_t =\mfalse$. Since they were in the same block, they will be in the same block in $\partit_{k+1}$. 
	
	Now consider the case that at least one of the states is able to reach $C$. Without loss of generality say that $s$ is able to reach $C$. Then there is a transition $s\rightarrow s'$ with $s'\in C$. By Definition~\ref{def:bisim}, there exists a $t'\in S$ such that $t\rightarrow t'$ and $s'\bisim_{\pi_0} t'$. By the induction hypothesis we know that since $s' \bisim_{\pi_0} t'$, $s'$ and $t'$ must be in the same block in $\partit_k$, i.e., $t'$ is in $C$. This witnesses that $t$ is also able to reach $C$ and we must have $\mathit{mark}_s = \mathit{mark}_t = \mtrue$. Since the states $s$ and $t$ are both marked and are in the same block in $\partit_k$, they will remain in the same block in $\partit_{k+1}$.
\end{proof}

\begin{lemma}\label{lemma:bisimulation}
	Let $S$ be the input set of states with $\rightarrow:S \times S$, $\blocklabel = S$ the block labels, and $\partit_n$ the partition stored in the memory after termination of Algorithm~\ref{lst:alg}. Then the relation induced by $\partit_n$ is a bisimulation relation.
\end{lemma}
\begin{proof}
	Since the program finished, we know that for all block indices $i\in \blocklabel$ we have $\mathit{unstable}_i = \mfalse$. For a block index $i \in \blocklabel$, $\mathit{unstable}_i$ is set to $\mfalse$ if the partition $\pi_k$, after iteration $k$, is stable under the block with index $i$ and set to $\mtrue$ if it is split. So, by Fact~\ref{fact:stable}, we know $\partit_n$ is stable under every block $B$, hence stable. Next, we prove that a stable partition is a bisimulation relation.
	
	We show that the relation $R$ induced by $\partit_n$ is a bisimulation relation. Assume states $s,t\in S$ with $sRt$ are in block $B\in \partit_n$. Consider a transition $s\rightarrow s'$ with $s'\in S$. State $s'$ is in some block $B' \in \partit_n$, and since the partition is stable under block $B'$, and $s$ is able to reach $B'$, by the definition of stability, we know that $t$ is also able to reach $B'$. Therefore, there must be a state $t'\in B'$ such that $t\rightarrow t'$ and $s'Rt'$. Finally, by the fact that $R$ is an equivalence relation we know that $R$ is also symmetric, therefore it is a bisimulation relation. 
\end{proof}

\begin{thm}\label{thm:correctness}
	The partition resulting from executing Algorithm~\ref{lst:alg} forms the coarsest relational partition for a set of states $S$ and a transition relation $\rightarrow: S \times S$, solving RCPP.  
\end{thm}
\begin{proof}
	By Lemma~\ref{lemma:bisimulation}, the resulting partition is a bisimulation relation. Lemma~\ref{lemma:invariant} implies that it is the coarsest refinement which is a bisimulation. 
\end{proof}
\subsection{Complexity analysis}
Every step in the body of the \textbf{while}-loop can be executed in constant time. So the asymptotic complexity of this algorithm is given by the number of iterations. 

\begin{thm}\label{thm:analysis}
	RCPP on an input with $m$ transitions and $n$ states is solved by Algorithm~\ref{lst:alg} in $\bigO(n)$ time using $max(m,n)$ CRCW PRAM processors.
\end{thm}
\begin{proof}
	In iteration $k \in \Nat$ of the algorithm, let us call the total number of blocks $N_k \in\Nat$ and the number of unstable blocks $U_k \in\Nat$. Initially, $N_0 = U_0 = |\pi_0|$. In every iteration $k$, a number of blocks $l_k \in\Nat$ is split, resulting in $l_k$ new blocks, so the new total number of blocks at the end of iteration $k$ is $ N_{k+1} = N_k +l_k $. 
	
	First the current block $C$ in iteration $k$ which was unstable is set to stable which causes the number of unstable blocks to decrease by one. In this iteration $k$ the $l_k$ blocks $B_1, \dots, B_{l_k}$ are split, resulting in $l_k$ newly created blocks. These $l_k$ blocks are all unstable. A number of blocks $l_k' \leq l_k$ of the blocks $B_1, \dots B_{l_k}$, were stable and are set to unstable again. The block $C$ which was set to stable is possibly one of these $l_k'$ blocks which were stable and set to unstable again. The total number of unstable blocks at the end of iteration $k$ is $U_{k+1} = U_{k} + l_k + l_k' - 1$.
	
	For all $k\in \Nat$, in iteration $k$ we calculate the total number of blocks $N_k = \sum_{i=0}^{k-1}(l_i)+|\pi_0|$ and unstable blocks $U_k = \sum_{i=0}^{k-1}(l_i + l_i') - k + |\pi_0|$. The number of iterations is given by $k = \sum_{i=0}^{k-1}(l_i + l_i') - U_k + |\pi_0|$. By definition, $l_i' \leq l_i$, and the total number of newly created blocks is $\sum_{i=0}^{k-1}(l_i) = N_k - |\pi_0|$, hence $\sum_{i=0}^{k-1}(l_i + l_i')\leq2\sum_{i=0}^{k-1}(l_i)\leq2N_k-2|\pi_0|$. The number of unstable blocks is always positive, i.e., $U_k \geq 0$, and the total number of blocks can never be larger than the number of states, i.e.,  $N_k \leq n$, so the total number of iterations $z$ is bounded by  $z \leq 2N_z - |\pi_0| \leq 2n - |\pi_0|$.
	
\end{proof}

\section{Bisimulation Coarsest Refinement Problem}\label{sec:labels}
In this section we extend our algorithm to the Bisimulation Coarsest Refinement Problem (\BCRP{}), i.e.,
to LTSs with multiple action labels. 

Solving \BCRP{} can in principle be done by translating an LTS to a Kripke structure, for instance by using the method described in~\cite{reniers2014results}. This translation introduces a new state for every transition,
resulting in a Kripke structure with $n+m$ states. 
If the number of transitions is significantly larger than the number of states, 
then the number of iterations of our algorithm increases undesirably.  


\subsection{The PRAM Algorithm}
Instead of introducing more states, we introduce multiple marks per state, but in total we have no more than $m$ marks. For each state $s$, we use a mark variable for each different outgoing action label relevant for $s$, i.e., for each $a$ for which there is a transition $s \pijl{a} s'$ to some state $s'$. Each state may have a different set of outgoing action labels and thus a different set of marks. Therefore, we first perform a preprocessing procedure in which we group together states that have the same set of outgoing action labels. This is valid, since two bisimilar states must have the same outgoing actions. That two states of the same block have the same set of action labels is then an invariant of the algorithm, since in the sequence of produced partitions, each partition is a refinement of the previous one. For the extended algorithm, we need to maintain extra information in addition to the information needed for Algorithm~\ref{lst:alg}. For an input LTS $A = (S, Act, \pijl{})$ with $n$ states and $m$ transitions this is the following extra information:
\begin{enumerate}
	\item Each action label has an index $a \in \{0, \dots,|Act| -1\}$.
	\item The following is stored in lists of size $m$, for each transition $s\pijl{a} t$ with transition index $i \in \{0, \dots, m-1\}$: 
	\begin{enumerate}
		\item $a_i := a$
		\item $\mathit{order}_i : \Nat$, the order of this action label, with respect to the source state $s$. E.g., if a state s has the list $[1, 3, 6]$ of outgoing action labels, and transition $i$ has label $3$, then $\mathit{order}_i$ is 1 (we start counting from $0$).
	\end{enumerate}
	\item $\mathit{mark} : [\Bool]$, a list of up to $m$ marks, in which there is a mark for every state, action pair for which it holds that the state has at least one outgoing transition labelled with that action. This list can be interpreted as the concatenation of lists $\mathit{mark}_s$ for all states $s \in S$. Essentially, we have for each state $s \in S$: 
	\begin{enumerate}
		\item $\markoffset(s) : \Nat$, the offset to access the marks of a given state $s$ in $\mathit{mark}$.
		\item $\mathit{mark}_{\markoffset(s)} : [ \Bool ]$, a list of marks (the list starting at position $\markoffset(s)$ in $\mathit{mark}$), where each mark indicates if the state can reach the current block with the corresponding action. We also refer to this list as $\mathit{mark}_s$. E.g., if state $s$ has actions $[1, 3, 6]$ and only actions $1$ and $6$ can reach the current block, this list has the contents $[\true, \false, \true]$.
		\item $\mathit{nr\_marks}_s$, the number of marks this state has, thus the length of list $\mathit{mark}_s$.
	\end{enumerate}
	\item $\mathit{mark\_length}$: The total length of all the $\mathit{mark}_s$ lists together, i.e., the sum of all the $\mathit{nr\_marks}_s$. This allows us to reset all marks in constant time using $\mathit{mark\_length}$ processors. This number is not larger than the number of transitions ($\mathit{mark\_length} \leq m$).
	\item In a list of size $n$, we store for each state $s \in S$ a variable $\mathit{split}_s$. This indicates if the state will be split off from its block.
\end{enumerate}

With this extra information, we can alter Algorithm \ref{lst:alg} to work with labels. The new version is given in Algorithm \ref{lst:alg-lab}. The changes involve the following:
\begin{enumerate}
	\item Lines \ref{alg:reset-mark-1}-\ref{alg:reset-mark-2}: Reset the $\mathit{mark}$ list.
	\item Line \ref{alg:reset-split}: Reset the $\mathit{split}$ list.
	\item Line \ref{alg:mark-order}: When marking the transitions, we do this for the correct action label, using $\mathit{order}_i$. Note the indexing into $\mathit{mark}$. It involves the offset for the state $\mathit{source}_i$, and $\mathit{order}_i$.
	\item Lines \ref{alg:new-mark-1}-\ref{alg:new-mark-2}: We tag a state to be splitted off when it differs for any action label from the block leader.
	\item Line \ref{alg:check-mark}: If a state was tagged to be splitted off in the previous step, it should split off from its leader.
	\item Line \ref{alg:unstable}: If any block was split, the partition may not be stable w.r.t.\ the splitter.
\end{enumerate}

\begin{algorithm}[!t]
	\SetAlgoLined
	\SetKwRepeat{Do}{do}{while}
	\If{$i < n$}{
		 $\mathit{unstable}_i := \mfalse$\;
		$\mathit{unstable}_{\mathit{block}_i} := \mtrue$\;
	}
	\Do{$C \neq\bot$}{
		$C := \bot$\;
		\textcolor{blue}{\If{$i < \mathit{mark\_length}$ \label{alg:reset-mark-1}} {
			$\textit{mark}_i := \false$\;
		}\label{alg:reset-mark-2}}
		\If{$i < n$ }{
			\textcolor{blue}{$split_i := \false$\;}\label{alg:reset-split}
			\If{$\mathit{unstable}_i$}
			{$C := i$\;}
		}\label{alg:reset-all}
		\If{$i < m$ and $\mathit{block}_{\mathit{target}_i} = C$} {\label{alg:if_reachable-1}
			\textcolor{blue}{$\mathit{mark}_{\markoffset({\mathit{source}_i})+\mathit{order}_i} := \true$\;}\label{alg:mark-order}
		}\label{alg:mark-order-done}
		\textcolor{blue}{\If{$i < m$ and $\mathit{mark}_{\markoffset({\mathit{source}_i})+\mathit{order}_i} \neq \mathit{mark}_{\markoffset(\mathit{block}_{\mathit{source}_i})+\mathit{order}_i}$ \label{alg:new-mark-1}} {
			$\mathit{split}_{\mathit{source}_i} := \true$\;
		}\label{alg:new-mark-2}}
		\If{$i < n$ and $C \neq \bot$}{
			$\mathit{unstable}_{C} := \mfalse$\;
			\If{\textcolor{blue}{$\mathit{split}_{i}$} \label{alg:check-mark}}{
				$\mathit{new\_leader}_{\mathit{block}_i} := i$\;\label{alg:leader-elect}
				$\mathit{unstable}_{\mathit{block}_i} := \mtrue$\;
				$\mathit{block}_i := \mathit{new\_leader}_{\mathit{block}_i}$\; 
				$\mathit{unstable}_{\mathit{block}_i} := \mtrue$\;		
				\textcolor{blue}{$\mathit{unstable}_C := \mtrue$\;}\label{alg:unstable}
			}
			
		}
	}
	\caption{\label{lst:alg-lab} The Algorithm for BCCP, the highlighted lines differ from Algorithm \ref{lst:alg}.}
\end{algorithm}

\begin{figure}[!t]
\centering
\begin{subfigure}{0.45\textwidth}
	\scalebox{1}{	
		\begin{tikzpicture}[->,shorten >=1pt,auto,node distance=1.25cm]
		
		\node[state] (A)                    {$s_0$};
		\node[state]         (B) [below left = 1.5cm and 1.5cm of A] {$s_1$};
		\node[state]         (C) [right = 3 cm of B] {$s_2$};
		
		\path (A) edge[bend left]  node {a} (B)
		(B) edge[bend left]  node {a} (A)
		(B) edge[bend right]  node {b} (C)
		
		(A) edge[bend left]  node {a,c} (C)
		(B) edge [loop below]  node {c} (B)
		(C) edge[bend left]  node {c} (A)
		(C) edge [loop below]  node {c} (C)
		;
		\end{tikzpicture}
	}
\end{subfigure}
\begin{subfigure}{0.45\textwidth}
	\begin{tabular}[b]{l|lll|lll|ll}
		$\mathit{source}_i$         & 0 & 0 & 0 & 1 & 1 & 1 & 2 & 2 \\
		$a_i$              & a & a & c & a & b & c & c & c \\ \hline
		$\mathit{action\_switch}_i$ & 0 & 0 & 1 & 0 & 1 & 1 & 0 & 0 \\ \hline
		$\mathit{order}_i$          & 0 & 0 & 1 & 0 & 1 & 2 & 0 & 0 \\
		$\mathit{nr\_mark}_s$       &   &   & 2 &   &   & 3 &   & 1 \\ \hline
		$\markoffset$    &   &   & 2 &   &   & 5 &   & 6
	\end{tabular}%
\end{subfigure}
	\caption{An example LTS and its derived preprocessing information.\label{fig:preprocessing}}
\end{figure}
\subsubsection{Preprocessing.}
To use the above algorithm, we need to do two preprocessing steps. First, we need to partition the states w.r.t.\ their set of outgoing action labels. This can be done with an altered version of Algorithm \ref{lst:alg}. Instead of splitting on a block at line \ref{alg:reachable}, we split on an action $a \in A$. We visit all transitions, and we mark the source if it has the same action label $a$. This can be found in Algorithm \ref{AlgAction}.

\begin{algorithm}[h]
	\setcounter{AlgoLine}{14}
	\SetAlgoLined
	\If{$i < m$ and $a_{i} = a$} {
		$\mathit{mark}_{\mathit{source_i}} := \true$\;
	}
	\caption{Marking the source per action label $a_i$.}
	\label{AlgAction}
\end{algorithm}

After executing Algorithm \ref{AlgAction}, each block can split in two blocks: a block that contains states that have $a$ as an outgoing action label and a block with states that do not have this outgoing action label. After doing this for all different action labels we end up with a partition of blocks, in which all states of a block have the same set of outgoing action labels, and each pair of states from different blocks have different sets of outgoing action labels. Using $m$ processors, this partition can be constructed in $\bigO(|Act|)$ time.

For the second preprocessing step, we need to gather the extra information that is needed in Algorithm~\ref{lst:alg-lab}. Only $a_i$ is part of the input, the others need to be calculated. We start our preprocessing by sorting the transitions by $(\mathit{source}_i, a_i)$, which can be done in $\bigO(\log m)$ time with $m$ processors, for instance using a parallel merge sort \cite{cole1988parallel}. In order to calculate $\mathit{order}_i$ and $\mathit{nr\_marks}_s$, we first calculate $\mathit{action\_switch}_i$ for each transition $i$, which is done in Algorithm~\ref{lst:preprocess}. See Figure \ref{fig:preprocessing} for an example. Now, $\mathit{order}_i$ can be calculated with a parallel segmented prefix sum~\cite{parallelscan} (also called a segmented scan) of $\mathit{action\_switch}$. A parallel segmented sum can be performed on $\mathit{action\_switch}$ to calculate $\mathit{nr\_marks}_s$, where we make sure to set $\mathit{nr\_marks}_s$ to $0$, if state $s$ has no outgoing transitions.
Finally, $\markoffset_s$, for the mark offsets, can be constructed as a list and calculated by applying a parallel prefix sum on $\mathit{nr\_marks}_s$. 
 The code in Algorithm~\ref{lst:preprocess} takes $\bigO(1)$ time on $m$ processors, and a parallel segmented (prefix) sum takes $\bigO(\log m)$ time~\cite{parallelscan}.

In total the preprocessing takes $\bigO(|Act| + \log m)$ time.

\begin{algorithm}[t]
	\SetAlgoLined
	\If{i $\leq$ m}{
		\uIf{ $i = 0$ or $\mathit{source}_i \neq \mathit{source}_{i-1}$ or $a_i = a_{i-1}$}{
			$\mathit{action\_switch}_i = 0$;
			}
		\Else{
			$\mathit{action\_switch}_i = 1;$
		}
	}
\caption{\label{lst:preprocess} Preprocessing step needed for Algorithm~\ref{lst:alg-lab}. We calculate $\mathit{action\_switch}_i$, which is needed for the $\mathit{order}_i$ and $\mathit{nr\_marks}_s$ variables.}
\end{algorithm}

\subsection{Complexity \& Correctness}\label{sec:lab-complex-correct}
For Algorithm \ref{lst:alg-lab}, we need to prove why it takes a linear number of steps to construct the final partition. This is subtle, as an iteration of the algorithm does not necessarily produce a stable block.

\begin{thm}
	Algorithm \ref{lst:alg-lab} on an input LTS with $n$ states and $m$ transitions will terminate in $\bigO(n + |Act|)$ steps.
\end{thm}
\begin{proof}
	The total preprocessing takes $\bigO(|Act| + \log m)$ steps, after which the \textbf{while}-loop will be executed on a partitioning $\pi_0$ which was the result of the preprocessing on the partition $\{S\}$. Every iteration of the \textbf{while}-loop is still executed in constant time. Using the structure of the proof of Theorem~\ref{thm:analysis}, we derive a bound on the number of iterations.
	
	At the start of iteration $k\in\Nat$ the total number of blocks and unstable blocks are $N_k,U_k\in \Nat$, initially $U_0 = N_0 = |\pi_0|$. In iteration $k$, a number $l_k$ of blocks is split in two blocks, resulting in $l_k$ new blocks, meaning that $N_{k+1} = N_{k} + l_k$. All new $l_k$ blocks are unstable and a number $l_k' \leq l_k$ of the old blocks that are split, were stable at the start of iteration $k$ and now unstable. If $l_k = l_k' = 0$ there are no blocks split and the current block $C$ becomes stable. We indicate this with a variable $c_k$: $c_k=1$ if $l_k = 0$, and $c_k = 0$, otherwise. The total number of iterations up to iteration $k$ in which no block is split is given by $\sum_{i=0}^{k-1} c_i$. The number of iterations in which at least one block is split is given by $k - \sum_{i=0}^{k-1} c_i$.
	
	If in an iteration $k$ at least one block is split, the total number of blocks at the end of iteration $k$ is strictly higher than at the beginning, hence for all $k\in \Nat$, $N_k \geq k - \sum_{i=0}^{k-1}c_i$. Hence, $N_k+\sum_{i=0}^{k-1}c_i$ is an upper bound for $k$.
	
	We derive an upper bound for the number of iterations in which no blocks are split using the total number of unstable blocks. In iteration $k$  there are $U_k = \sum_{i=0}^{k-1}(l_i + l_i') - \sum_{i=0}^{k-1} c_i + |\pi_0|$ unstable blocks. Since the sum of newly created blocks $\sum_{i=0}^{k-1}(l_i) = N_k-|\pi_0|$ and $l_i' \leq l_i$, the number of unstable blocks $U_k$ is bounded by $2N_k-\sum_{i=0}^{k-1}c_i - |\pi_0|$. Since $U_k\geq 0$ we have the bound $\sum_{i=0}^{k-1}c_i\leq2N_k-|\pi_0|$. This gives the bound on the total number of iterations $z \leq 3N_z-|\pi_0| \leq 3n - |\pi_0|$. 
	
	With the time for preprocessing this makes the total run time complexity $\bigO(n + |Act| + \log m)$. Since the total number of transitions $m$ is bounded by $|Act| \times n^2$, this simplifies to $\bigO(n + |Act|)$.
\end{proof}

Concerning correctness, we need to address two things. Firstly, as argued above, we start with a different partition compared to Algorithm~\ref{lst:alg}, but it is a valid choice since states with different outgoing labels can never be bisimilar. Secondly, although the partition may not become stable w.r.t.\ the splitter, this will eventually occur, and the algorithm will only stop once the partition is stable w.r.t.\ all blocks.  Therefore, the algorithm will produce the coarsest bisimulation relation.

\section{Experimental Results}\label{sec:experiments}
In order to validate the proposed algorithm, we implemented Algorithm \ref{lst:alg-lab} from Section \ref{sec:labels}. The implementation targets graphics processing units (GPUs) since a GPU closely resembles a PRAM and supports a large amount of parallelism. The algorithm is implemented in CUDA version 11.1 with use of the Thrust library.\footnote{The source code can be found at https://github.com/sakehl/gpu-bisimulation.} As input, we chose all benchmarks of the VLTS benchmark suite\footnote{https://cadp.inria.fr/resources/vlts/.} for which the implementation produced a result within 10 minutes. The VLTS benchmarks are LTSs that have been derived from real concurrent system models. 

The experiments were run on an  NVIDIA Titan RTX
with 24 GB memory and 72 Streaming Multiprocessors, each supporting up to 1,024 threads in flight.  Although this GPU supports 73,728 threads in flight, it is very common to launch a GPU program with one or even several orders of magnitude more threads, in particular to achieve load balancing between the Streaming Multiprocessors and to hide memory latencies. In fact, the performance of a GPU program usually relies on that many threads being launched.

Our implementation is purely a proof of concept, to show that our algorithm can be 
mapped to actual hardware and to understand how the algorithm 
scales with the number of states and transitions.


\begin{table*}[tb]
	\centering
	
	\resizebox{\textwidth}{!}{%
\begin{tabular}{l|r|r|r|r|r|r|r|r|r|r|r|r}
	\hline
	Benchmark name & 
	\multicolumn{1}{l|}{$n$} & 
	\multicolumn{1}{l|}{$m$} & 
	\multicolumn{1}{l|}{$|\mathit{Act}|$} & 
	\multicolumn{1}{l|}{$|\mathit{Blocks}|$} & 
	\multicolumn{1}{l|}{$\#\mathit{It}$} & 
	\multicolumn{1}{l|}{$T_\mathit{pre}$} &
	\multicolumn{1}{l|}{$T_\mathit{alg}$} &
	\multicolumn{1}{l|}{$T_\mathit{total}$} & 
	$\#\mathit{It}/n$ & 
	\begin{tabular}[c]{@{}l@{}}$\#\mathit{It}/$\\ $|\mathit{Blocks}|$\end{tabular} & 
	$T_\mathit{total}/n$ & 
	$T_\mathit{alg}/\#\mathit{It}$ \\ 
	\hline
Vasy\_0\_1        & 289       & 1,224      & 2      & 9        & 16      & 0.50      & 0.37      & 0.87        & 0.06   & 1.78                                                     & 0.003         & 0.023          \\
Cwi\_1\_2         & 1,952     & 2,387      & 26     & 1,132    & 2,786   & 0.63      & 56.5      & 57.1        & 1.43   & 2.46                                                     & 0.029         & 0.020          \\
Vasy\_1\_4        & 1,183     & 4,464      & 6      & 28       & 45      & 0.56      & 1.01      & 1.58        & 0.04   & 1.61                                                     & 0.001         & 0.022          \\
Cwi\_3\_14        & 3,996     & 14,552     & 2      & 62       & 122     & 0.63      & 2.68      & 3.30        & 0.03   & 1.97                                                     & 0.001         & 0.022          \\
Vasy\_5\_9        & 5,486     & 9,676      & 31     & 145      & 193     & 0.84      & 4.22      & 5.06        & 0.04   & 1.33                                                     & 0.001         & 0.022          \\
Vasy\_8\_24       & 8,879     & 24,411     & 11     & 416      & 664     & 0.70      & 13.9      & 15          & 0.07   & 1.59                                                     & 0.002         & 0.021          \\
Vasy\_8\_38       & 8,921     & 38,424     & 81     & 219      & 319     & 1.12      & 6.64      & 7.76        & 0.04   & 1.46                                                     & 0.001         & 0.021          \\
Vasy\_10\_56      & 10,849    & 56,156     & 12     & 2,112    & 3,970   & 0.73      & 82.0      & 82.7        & 0.37   & 1.88                                                     & 0.008         & 0.021          \\
Vasy\_18\_73      & 18,746    & 73,043     & 17     & 4,087    & 6,882   & 1.01      & 142       & 143         & 0.37   & 1.68                                                     & 0.008         & 0.021          \\
Vasy\_25\_25      & 25,217    & 25,216     & 25,216 & 25,217   & 25,218  & 159       & 519       & 678         & 1.00   & 1.00                                                     & 0.027         & 0.021          \\
Vasy\_40\_60      & 40,006    & 60,007     & 3      & 40,006   & 87,823  & 0.87      & 1,810     & 1,811       & 2.20   & 2.20                                                     & 0.045         & 0.021          \\
Vasy\_52\_318     & 52,268    & 318,126    & 17     & 8,142    & 15,985  & 2.52      & 338       & 340         & 0.31   & 1.96                                                     & 0.007         & 0.021          \\
Vasy\_65\_2621    & 65,537    & 2,621,480  & 72     & 65,536   & 98,730  & 12.2      & 10,050    & 10,060      & 1.51   & 1.51                                                     & 0.154         & 0.102          \\
Vasy\_66\_1302    & 66,929    & 1,302,664  & 81     & 66,929   & 91,120  & 6.70      & 5,745     & 5,752       & 1.36   & 1.36                                                     & 0.086         & 0.063          \\
Vasy\_69\_520     & 69,754    & 520,633    & 135    & 69,754   & 113,246 & 4.13      & 3,780     & 3,780       & 1.62   & 1.62                                                     & 0.054         & 0.033          \\
Vasy\_83\_325     & 83,436    & 325,584    & 211    & 83,436   & 148,012 & 4.41      & 3,093     & 3,097       & 1.77   & 1.77                                                     & 0.037         & 0.021          \\
Vasy\_116\_368    & 116,456   & 368,569    & 21     & 116,456  & 210,537 & 2.50      & 5,900     & 5,900       & 1.81   & 1.81                                                     & 0.051         & 0.028          \\
Cwi\_142\_925     & 142,472   & 925,429    & 7      & 3,410    & 5,118   & 4.85      & 238       & 243         & 0.04   & 1.50                                                     & 0.002         & 0.047          \\
Vasy\_157\_297    & 157,604   & 297,000    & 235    & 4,289    & 9,682   & 4.58      & 201       & 206         & 0.06   & 2.26                                                     & 0.001         & 0.021          \\
Vasy\_164\_1619   & 164,865   & 1,619,204  & 37     & 1,136    & 1,630   & 8.34      & 125       & 134         & 0.01   & 1.43                                                     & 0.001         & 0.077          \\
Vasy\_166\_651    & 166,464   & 651,168    & 211    & 83,436   & 145,029 & 6.13      & 5,710     & 5,720       & 0.87   & 1.74                                                     & 0.034         & 0.039          \\
Cwi\_214\_684     & 214,202   & 684,419    & 5      & 77,292   & 149,198 & 3.58      & 6,948     & 6,952       & 0.70   & 1.93                                                     & 0.032         & 0.047          \\
Cwi\_371\_641     & 371,804   & 641,565    & 61     & 33,994   & 85,858  & 4.72      & 4,050     & 4,050       & 0.23   & 2.53                                                     & 0.011         & 0.047          \\
Vasy\_386\_1171   & 386,496   & 1,171,872  & 73     & 113      & 199     & 7.38      & 14.0      & 21          & 0.00   & 1.76                                                     & 0.000         & 0.070          \\
Cwi\_566\_3984    & 566,640   & 3,984,157  & 11     & 15,518   & 23,774  & 16.0      & 3,707     & 3,723       & 0.04   & 1.53                                                     & 0.007         & 0.156          \\
Vasy\_574\_13561  & 574,057   & 13,561,040 & 141    & 3,577    & 5,860   & 71.5      & 3,770     & 3,841       & 0.01   & 1.64                                                     & 0.007         & 0.643          \\
Vasy\_720\_390    & 720,247   & 390,999    & 49     & 3,292    & 3,782   & 3.97      & 143       & 147         & 0.01   & 1.15                                                     & 0.0002        & 0.038          \\
Vasy\_1112\_5290  & 1,112,490 & 5,290,860  & 23     & 265      & 365     & 24.0      & 99.3      & 123         & 0.0003 & 1.38                                                     & 0.0001        & 0.272          \\
Cwi\_2165\_8723   & 2,165,446 & 8,723,465  & 26     & 31,906   & 66,132  & 37.0      & 23,660    & 23,700      & 0.03   & 2.07                                                     & 0.011         & 0.358          \\
Cwi\_2416\_17605  & 2,416,632 & 17,605,592 & 15     & 95,610   & 152,099 & 64.1      & 96,400    & 96,500      & 0.06   & 1.59                                                     & 0.040         & 0.634          \\
Vasy\_6020\_19353 & 6,020,550 & 19,353,474 & 511    & 7,168    & 12,262  & 221       & 11,690    & 11,910      & 0.002  & 1.71                                                     & 0.002         & 0.954          \\
Vasy\_6120\_11031 & 6,120,718 & 11,031,292 & 125    & 5,199    & 10,014  & 74.0      & 6,763     & 6,837       & 0.002  & 1.93                                                     & 0.001         & 0.675          \\
Vasy\_8082\_42933 & 8,082,905 & 42,933,110 & 211    & 408      & 660     & 281       & 1,149     & 1,429       & 0.0001 & 1.62                                                     & 0.0002        & 1.739                                    
\end{tabular}%
}
\caption{\label{tab:results}Benchmark results for  Algorithm \ref{lst:alg-lab} on a GPU, times ($T$) are in ms. }
	
\end{table*}

In the implementation, we have to make a few adjustments, since a GPU differs in some aspects from a PRAM. 
To make memory updates globally visible, we need to synchronize at certain points of Algorithm \ref{lst:alg-lab}, otherwise the changes in the memory are not consistent. We do this by splitting up the algorithm in different kernels (functions that execute in parallel on a GPU) since after a kernel run all processors (threads) are synchronized.

To be precise, in Algorithm \ref{lst:alg-lab} we need to synchronize after:
\begin{itemize}
	\item \textbf{Line \ref{alg:reset-all}}: To make sure the $\mathit{mark}$ and $\mathit{split}$ lists
	are reset and the splitter ($C$) is the same for all threads.
	\item \textbf{Line \ref{alg:mark-order-done}}: To make sure every thread has the same view of the $\mathit{mark}$ list.
	\item \textbf{Line \ref{alg:new-mark-2}}: To synchronize the $\mathit{mark}$ list.
	\item \textbf{Line \ref{alg:leader-elect}}: To make sure the next leader for states that split off ($\mathit{new\_leader}_{\mathit{block}_i}$) is chosen consistently among threads.
\end{itemize}
We have chosen to allow race conditions in our implementation, for instance at Line 6 where multiple blocks can mark themselves as current ($C$). Strictly speaking, this is not safe in the CUDA programming model, but it does work for 32 bit words. This can be easily adjusted using atomic instructions, although this will result in sequentializing write accesses to the same memory location, meaning that a write need not be in constant time anymore.

To ensure the implementation also works when $n$ and/or $m$ is larger than the number of threads $d$ on the GPU, we encapsulate the \textbf{if}-\textbf{then} blocks at lines 1-4, 7-9, 10-15, 16-18, 19-21 and 22-31 of Algorithm~\ref{lst:alg-lab} each in a \textbf{for}-loop, in which every thread accesses not only the data elements associated with its global index $i$, but also, if needed, the elements with index $i+d$, $i+2d$, etc., as long as the indices are valid.

Table \ref{tab:results} shows the results of the experiments we conducted. The $|\mathit{Blocks}|$ column indicates the number of different blocks at the end of the algorithm, where each block contains only bisimilar states. With $\#\mathit{It}$
we refer to the number of \textbf{while}-loop iterations that were executed (see Algorithm \ref{lst:alg-lab}), 
before all blocks became stable. The $T_\mathit{pre}$ give the preprocessing times in milliseconds, which includes doing the memory transfers to the GPU, sorting the transitions and partitioning. The $T_\mathit{alg}$ give the times of the core algorithm, in milliseconds. The $T_\mathit{total}$ is the sum of the preprocessing and the algorithm, in milliseconds. We have not included the loading times for the files and the first CUDA API call that initializes the device.
We ran each benchmark 10 times and took the averages. The standard 
deviation of the total times varied between 0\% and 3\% of the average, thus 10 runs are sufficient. All the times are rounded with respect to the standard error of the mean.

We see that the bound as proven in Section \ref{sec:lab-complex-correct} ($k \leq 3n$) is indeed respected, $\#\mathit{It}/n$ is at most 2.20, and most of the time below that. The number of iterations is tightly related to the amount of blocks that the final partition has, the $\#\mathit{It} / |\mathit{Blocks}|$ column varies between 1.00 and 2.53. This can be understood by the fact that each iteration either splits one or more blocks or marks a block as stable, and all blocks must be checked on stability at least once. This also means that for certain LTSs the algorithm scales better than linearly in $n$. The preprocessing often takes the same amount of time (about a few milliseconds). Exceptions are those cases with a large number of actions and/or transitions.

Concerning the GPU run times, it is not true that each iteration takes the same amount of time. A GPU is not a perfect PRAM machine. There are two key differences. Firstly, we suspect that the algorithm is memory bound since it is performing a limited amount of computations. The memory accesses are irregular, i.e., random, which caches can partially compensate, but for sufficiently large $n$ and $m$, the caches cannot contain all the data. This means that as the LTSs become larger, memory accesses become relatively slower. Secondly, at a certain moment, the maximum number of threads that a GPU can run in parallel is achieved, and adding more threads will mean more run time. These two effects can best be seen in the column headed by $T_\mathit{alg}/\#\mathit{It}$, which corresponds to the time per iteration. The values are around $0.02$ up to $300,000$ transitions, but for a higher number of states and transitions, the amount of time per iteration increases.

\subsection{Experimental comparison}
\begin{figure*}[!tb]
	\centering
    \begin{subfigure}{0.5\textwidth} 
    	\centering
		\includegraphics[scale=0.45]{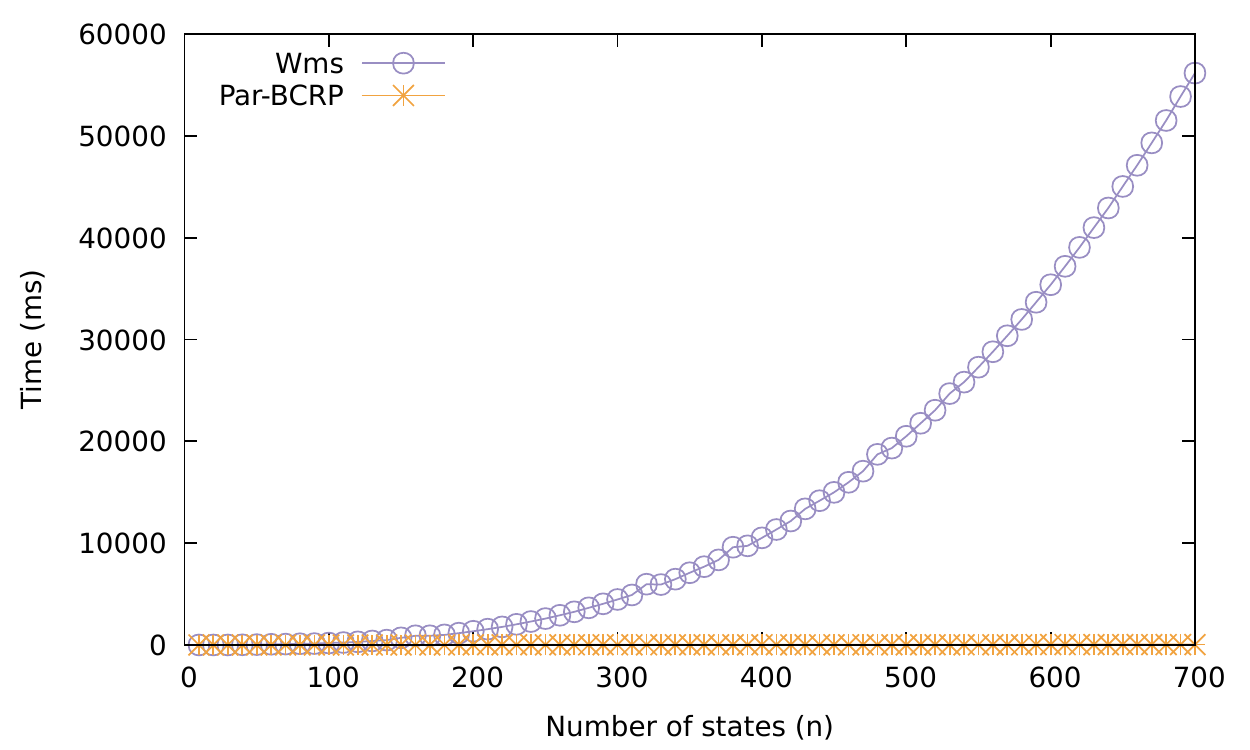}	
	\end{subfigure}%
 	\begin{subfigure}{0.5\textwidth}
 		\centering
 		\scalebox{0.75}{
		\begin{tabular}{l|ll}
			States         & Run time Wms (ms) & Par-BCRP (ms) \\ \hline
			10 & 1 & 1 \\
			100 & 182 & 5\\
			200 & 1350 & 9\\
			300 & 4463 & 13\\
			400 & 10519 & 18\\
			500 & 20508 & 22\\
			600 & 35392 & 26\\
			700 & 56183 & 30\\
					\multicolumn{3}{c}{}\vspace{1ex}\\
		\end{tabular}}
	\end{subfigure}%
	\caption{Run times of Par-BCRP on the LTS $\mathit{Fan\_out}_n$.\label{fig:slow_results}}
\end{figure*}
We compared our implementation (Par-BCRP) with an implementation of the algorithm by Lee and Rajasekaran (LR)~\cite{lee_1994_RCPP} on GPUs, and the optimized GPU implementation by Wijs based on \emph{signature-based} bisimilarity checking~\cite{blom2003distributed}, with \textit{multi-way splitting} (Wms) and with \textit{single-way splitting} (Wss)~\cite{wijs_2015}.  Multi-way splitting indicates that a block is split in multiple blocks at once, which is achieved in signature-based algorithms by computing a signature for each state in every partition refinement iteration, and splitting each block off into sets of states, each containing all the states with the same signature. The signature of a state is derived from the labels of the blocks that this state can reach in the current partition.


\begin{table}
	\centering
	
\begin{tabular}{l|r|r|r|r}
	\hline
Benchmark name    & LR      & Wms        & Wss          & Par-BCRP \\ \hline
Vasy\_0\_1        & 2.29    & 0.45       & 0.49         & 0.87                                                             \\
Cwi\_1\_2         & 17      & 21.8       & 18.8         & 57.1                                                             \\
Vasy\_1\_4        & 4.78    & 0.62       & 1.68         & 1.58                                                             \\
Cwi\_3\_14        & 60      & 3.72       & 3.80         & 3.30                                                             \\
Vasy\_5\_9        & 134     & 3.45       & 35.3         & 5.06                                                             \\
Vasy\_8\_24       & 277     & 3.03       & 31.5         & 15                                                               \\
Vasy\_8\_38       & 127     & 5.94       & 35.1         & 7.76                                                             \\
Vasy\_10\_56      & 860     & 4.6(0.2)   & 40.9         & 82.7                                                             \\
Vasy\_18\_73      & 1,354   & 21.7       & 211          & 143                                                              \\
Vasy\_25\_25      & 21,960  & 416        & t.o.         & 678                                                              \\
Vasy\_40\_60      & 17,710  & 1,230      & 1,290        & 1,811                                                            \\
Vasy\_52\_318     & 11,855  & 152(20)    & 368          & 340                                                              \\
Vasy\_65\_2621    & t.o.    & 1,230      & 27,000       & 10,060                                                           \\
Vasy\_66\_1302    & 480,600 & 240(20)    & 20,450       & 5,752                                                            \\
Vasy\_69\_520     & 94,800  & 35.4       & 16,090       & 3,780                                                            \\
Vasy\_83\_325     & 57,190  & 5,880      & 21,500       & 3,097                                                            \\
Vasy\_116\_368    & 80,900  & 2,930      & 6,360        & 5,900                                                            \\
Cwi\_142\_925     & 3,363   & 140(20)    & 220(30)      & 243                                                              \\
Vasy\_157\_297    & 1,058   & 579        & 1,240        & 206                                                              \\
Vasy\_164\_1619   & 8,173   & 46.8       & 470(30)      & 134                                                              \\
Vasy\_166\_651    & 80,210  & 9,560      & 29,660       & 5,720                                                            \\
Cwi\_214\_684     & 19,250  & 450(50)    & 440(30)      & 6,952                                                            \\
Cwi\_371\_641     & 26,940  & 1,548      & 6,970        & 4,050                                                            \\
Vasy\_386\_1171   & 334     & 34.8       & 30.6         & 21                                                               \\
Cwi\_566\_3984    & 98,200  & 2,200(200) & 6,700        & 3,723                                                            \\
Vasy\_574\_13561  & 144,810 & 1,853      & 11,700       & 3,841                                                            \\
Vasy\_720\_390    & 2,454   & 183        & 1,633        & 147                                                              \\
Vasy\_1112\_5290  & 4,570   & 36.8       & 293          & 123                                                              \\
Cwi\_2165\_8723   & 140,170 & 1,965      & 9,700        & 23,700                                                           \\
Cwi\_2416\_17605  & 257,200 & 15,300     & 16,300(1100) & 96,500                                                           \\
Vasy\_6020\_19353 & 107,900 & 19,230     & 34,000(2000) & 11,910                                                           \\
Vasy\_6120\_11031 & 55,750  & 1,280      & 7,010        & 6,837                                                            \\
Vasy\_8082\_42933 & 17,272  & 2,030      & 5,530        & 1,429                                                                                   
\end{tabular}%
\caption{\label{tab:results-compare}Comparison of the different algorithms with times in ms.}

\end{table}

The running times of the different algorithms can be found in Table~\ref{tab:results-compare}. Similarly to our previous benchmarks, the algorithms were run 10 times on the same machine using the same VLTS benchmark suite with a time-out of 10 minutes. In some cases, the non-deterministic behaviour of the algorithms Wms and Wss led to high variations in the runs. In cases where the standard error of the mean was more than 5\% of the mean value, we have added the standard error in Table~\ref{tab:results-compare} in between parentheses. Furthermore, all the results are rounded with respect to the standard error of the mean.
As a pre-processing step for the LR, Wms and Wss algorithms the input LTSs need to be sorted.
We did not include this in the times, nor the reading of files and the first CUDA API call (which initializes the GPU).

This comparison confirms the expectation that our algorithm in all cases (except one small LTS) out-performs LR. This confirms our expectation that LR is not suitable for massive parallel devices such as GPUs. 
Furthermore, the comparison teaches that in most cases our algorithm (Par-BCRP) outperforms Wss. In some benchmarks (Cwi\_1\_2, Cwi\_214\_684, Cwi\_2165\_8723 and Cwi\_2416\_17605) Wss is more than twice as fast, but in 16 other cases our algorithm is more than twice as fast.
The last comparison shows us that our algorithm does not out-perform Wms. Wms employs multi-way splitting which is known to be very effective in practice. 
Contrary to our implementation, Wms is highly optimized for GPUs while the focus of the current work 
is to improve the theoretical bounds and describe a general algorithm.

In order to understand the difference between Wms and our algorithm better, we analysed the complexity of Wms \cite{wijs_2015}. In general this algorithm is quadratic in time, and 
the linearity
claim in \cite{wijs_2015} depends on the assumption that the fan-out of `practical' transition
systems is bounded, i.e., every state has no more than $c$ outgoing transitions for $c$ a
(low) constant.

We designed the transition systems $\textit{Fan\_out}_n$ for $n\in\Nat^+$ to illustrate the difference. 
The LTS $\textit{Fan\_out}_n = \mbox{$(S, \{a,b\}, \pijl{})$}$ has $n$ states: $S=\{0,\dots,n-1\}$. The transition function contains $i \pijl{a} i+1$ for all states $1 <i< n-1$. Additionally, 
from state $0$ and $1$ there are transitions to every state: $0\pijl{b} i, 1\pijl{b} i$ for all $i\in S$. This LTS has $n$ states, $3n-3$ transitions and a maximum out degree of $n$ transitions.  

Figure~\ref{fig:slow_results} shows the results of calculating the bisimulation equivalence classes for $\textit{Fan\_out}_n$, with Wms and Par-BCRP. It is clear that the run time for Wms increases quadratically as the number of states grows linearly,
already becoming untenable 
for a small amount of states. 
On the other hand, in conformance with our analysis, our algorithm scales linearly. 



\section{Weaker PRAM models}
\label{sec:weaker}
Algorithm~\ref{lst:alg} relies on concurrent writes to perform the constant time leader election and the choice of splitter. This means that the algorithm does not work on a weaker PRAM model. In this section we describe a modification for the \textit{common} CRCW PRAM and a limitation for the ERCW PRAM. 

It is shown in~\cite{kuvcera1982parallel} that any \textit{priority} CRCW PRAM using $n$ processors and $m$ memory cells can be simulated by a \textit{common} CRCW PRAM with $\bigO(n^2)$ processors and $\bigO(m^2)$ memory cells. For our problem, a \textit{common} CRCW PRAM with $\bigO(n^2)$ processors and no extra memory can solve leader election. 

This leader election on the \textit{common} CRCW PRAM is given in Algorithm~\ref{lst:alg_common}. Every processor is indexed as $P_{i,j}$ for all $i,j\in \{0, \dots, n -1\}$ for exactly $n^2$ processors. First, if $P_{i,j}$ has a state with index $i$ that is eligible to be the leader of a new block (line 1), it writes $\mathit{block}_i$, i.e., the index of the block the state is currently a member of, to position $i$ in a list $\mathit{new\_leader}$. In the next step, $P_{i,j}$ replaces $\mathit{new\_leader}_j$ with $0$ if $\mathit{new\_leader}_i = \mathit{new\_leader_j}$ and $i < j$. In other words, if $P_{i,j}$ encounters two states that can become the new leader, it selects the one with the smallest index. This is possibly a concurrent write, but all writes involve the same value $0$, hence this is allowed by the common CRCW PRAM. Next, if for $P_{i,j}$, $\mathit{new\_leader}_i \neq 0$, it writes the value $i$ to $\mathit{new\_leader}_{\mathit{block}_i}$ at line 7. For a given block $\mathit{block}_i$, the condition at line 6 only holds for the state with the largest index among the states that are split from $\mathit{block}_i$, hence there is at most one value is written.

\begin{algorithm}[t!]
	\If{$\mathit{mark}_i \neq \mathit{mark}_{\mathit{block}_i}$\label{alg:common_shouldsplit}}{
		$\mathit{new\_leader}_i := \mathit{block}_i$\;
		\If{$i < j$ and $\mathit{new\_leader}_i = \mathit{new\_leader}_j$} {
			$\mathit{new\_leader}_j := 0$\;
		}
		\If{$\mathit{new\_leader}_i \neq 0$} {
			$\mathit{new\_leader}_{\mathit{block}_i} := i$\;\label{alg:common_leader}
		}
		$\mathit{unstable}_{\mathit{block}_i} := \mtrue$\;
		$\mathit{block}_i := \mathit{new\_leader}_{\mathit{block}_i}$\; 
		$\mathit{unstable}_{\mathit{block}_i} := \mtrue$\;						
	}	
 	\caption{\label{lst:alg_common}Leader election for a \textit{common} CRCW PRAM}
\end{algorithm}

Leader election on the ERCW PRAM is not possible in constant time, which follows from a result by Cook et al.~\cite[Thm 4.]{cook1986upper}. This result says that all functions that have a \textit{critical} input are in $\Omega(log~n)$ on ERCW PRAMs. A bit sequence $I$ of size $n$ is critical for a function $f: \{ 0, 1\}^n \to \{ 0, 1 \}$ iff for any $I'$ obtained by flipping exactly one bit in $I$ we have $f(I) \neq f(I')$. Leader election can be seen as a function $f: S\rightarrow \{0,1\}$, where $f(i) = 1$ iff $i$ is elected as a new leader. This function has a critical input, namely the to be chosen leader.

\section{Related work}\label{sec:relatedwork}
In~\cite{lee_1994_RCPP} Lee and Rajasekaran study RCPP. They implement a parallel version of Kanellakis-Smolka that runs in $\bigO(n~log~n)$ time on $\frac{m}{log\ n}log\ log\ n$ CRCW PRAM processors. In~{\cite{rajasekaran1998parallel} they present a different algorithm based on Paige and Tarjan's algorithm~\cite{paige1987three} that has the same run time of $\bigO(n~log~n)$ but using only $\frac{m}{n}log~n$ CREW processors. Jeong et al.~\cite{jeong1998} presented a linear time parallel algorithm, but it is probablistic in the sense that it has a non-zero chance to output the wrong result. Furthermore, Wijs~\cite{wijs_2015} presented a GPU implementation of an algorithm to solve the strong and branching bisimulation partition refinement problems 
but although efficient for many practical cases, it has a quadratic time complexity. 

In a distributed setting, Blom and Orzan studied algorithms for refinement~\cite{blom2003distributed}. Those algorithms use message passing as ways of communication between different workers in a network and rely on a small number of 
processors. Therefore, they are very different in nature than our algorithm. 
Those algorithms were extended and optimized for branching bisimulation~\cite{blom2009}.

\section{Conclusion}
We proposed and implemented an algorithm for RCPP and BCPP. We proved that the algorithm stops in $\bigO(n + |Act|)$ steps on $max(n,m)$ CRCW PRAM processors. We implemented the algorithm for BCPP in CUDA, and conducted experiments that show the potential to compute bisimulation in practice in linear time. Further advances in parallel hardware will make this more feasible.

For future work, it is interesting to investigate whether RCPP can be solved in sublinear time, that is $\bigO(n^{\epsilon})$ for a $\epsilon < 1$, as requested in~\cite{lee_1994_RCPP}. 
It is also intriguing whether the practical effectiveness of the algorithm in \cite{wijs_2015}
by splitting blocks simultaneously can be combined with our algorithm, 
while preserving the linear time upperbound. 
Furthermore, it remains an open question whether our algorithm can be generalised for weaker bisimulations, such as weak and branching bisimulation~\cite{GlaWei96,jansen_2020}. The main challenge here is that the transitive closure of so-called internal steps needs to be taken into account.

\subsubsection*{Acknowledgments}
This work is carried out in the context of the NWO AVVA project 612.001751 and the NWO TTW ChEOPS project 17249.

\bibliography{main}

\end{document}